\documentclass[a4paper,10pt]{article}
\pdfoutput=1
\usepackage{amsmath}
\usepackage{verbatim}
\usepackage[usenames]{color}
\usepackage{graphicx}
\usepackage[all]{draftcopy}
\usepackage{type1cm}
\usepackage{listings}
\usepackage{eso-pic}
\numberwithin{equation}{section}
\newcommand{\defeq}{\stackrel{{\mathrm d \mathrm e \mathrm f}}{=}}
\newcommand{\ket}[1]{\left| #1\right\rangle}

\newcommand{\scalarnorm}[1]{\left|#1\right|}
\newcommand{\combn}[3]{\left(\begin{array}{c}#1\\#2\\#3\end{array}\right)}
\newtheorem{theorem}{Theorem}[section]
\newtheorem{lemma}[theorem]{Lemma}

\definecolor{Corcra}{rgb}{0.4,0.0,0.4}


\newenvironment{proof}[1][Proof]{\begin{trivlist}
\item[\hskip \labelsep {\bfseries #1}]}{\end{trivlist}}

\newenvironment{claim}[1][Claim]{\begin{trivlist}
\item[\hskip \labelsep {\bfseries #1}]}{\end{trivlist}}

\newcommand{\qed}{\nobreak \ifvmode \relax \else
      \ifdim\lastskip<1.5em \hskip-\lastskip
      \hskip1.5em plus0em minus0.5em \fi \nobreak
      \vrule height0.75em width0.5em depth0.25em\fi}

\title{One Dimensional Quantum Walks with Memory}
\author{Michael Mc Gettrick\thanks{The De Br\'un Centre for
Computational Algebra, School of Mathematics,
The National University of Ireland, Galway. {\ttfamily
michael.mcgettrick@nuigalway.ie}}}

\begin{document}

\maketitle

\begin{abstract}
We investigate the quantum versions of a one-dimensional 
random walk, whose corresponding Markov Chain is of order 2. 
This corresponds to the walk having a memory of up to 
two previous steps. We derive the amplitudes and probabilities
for these walks, and point out how they differ from both
classical random walks, and quantum walks without memory.

\end{abstract}

\section{Introduction}
``Standard'' One Dimensional Discrete Quantum Walks (also known as Quantum Markov Chains) take place 
on the state space spanned by vectors
\begin{equation}
\ket{n,p}
\end{equation}
where $n\in Z$ (the integers)
and $p\in \{0,1\}$ is a boolean variable (see \cite{venegas,konno}
for a comprehensive treatment). The second variable
$p$ is often called the `coin' state or the chirality, with 
0 representing spin up and 1 representing spin down. It
is the quantum part of the walk, while $n$ is the classical 
part. One step of the walk is given by the transitions
\begin{eqnarray}
\ket{n,0} &\longrightarrow a\ket{n-1,0} + b\ket{n+1,1}\\
\ket{n,1} &\longrightarrow c\ket{n-1,0} + d\ket{n+1,1}
\end{eqnarray} 
where 
\begin{equation}
\left( 
\begin{array}{ll}
 a & b \\ 
c & d
\end{array}
\right) 
\in SU(2),
\end{equation}
the group of $2 \times 2$ unitary matrices of determinant 1.
These walks have been well studied, and their asymptotic behaviour well analyzed
\cite{meyer,nayak,kempe,watrous,ambainis}.

The corresponding classical walk is represented by a Markov
Chain whose transition matrix is tridiagonal, with zeroes
along the diagonal, and 1/2 along off-diagonal (for the fair
coin):
\begin{equation}
\left( 
\begin{array}{llllllll}
 0 & 1/2 & 0 & \dots & \dots & \dots & \dots & 0 \\ 
 1/2 & \ddots & \ddots & \dots & \dots & \dots & \dots & \vdots \\ 
 0 & \ddots & \ddots & \ddots & \dots & \dots & \dots& \vdots \\ 
 \vdots & \dots & 1/2 & 0 & 1/2 & \dots & \dots& \vdots \\ 
 \vdots & \dots & \dots & 1/2 & 0 & 1/2 & \dots & \vdots \\ 
 \vdots & \dots & \dots & \dots & \ddots & \ddots & \ddots & 0 \\ 
 \vdots & \dots & \dots & \dots & \dots & \ddots & \ddots & 1/2 \\ 
 0 & \dots & \dots & \dots & \dots & 0 & 1/2 & 0 \\ 
\end{array}
\right) 
\end{equation}

In this paper we investigate quantum walks ``with memory'': The state space
is spanned by vectors of the form
\begin{equation}
\ket{n_r, n_{r-1}, \dots , n_2, n_1, p}
\end{equation}
where $n_j = n_{j-1} \pm 1$, since the walk only takes one step right or 
left at each time interval. $n_j$ is the position of the walk at time 
$t-j+1$ (so $n_1$ is the current position). The transitions are of the form
\begin{align}
\ket{n_r, n_{r-1}, \dots , n_2, n_1, 0} \longrightarrow
a&\ket{n_{r-1}, \dots , n_2, n_1, n_1 \pm 1,  0}\nonumber\\
&  +
b\ket{n_{r-1}, \dots , n_2, n_1, n_1 \pm 1,  1}\label{eq:tr1}\\
\ket{n_r, n_{r-1}, \dots , n_2, n_1, 1} \longrightarrow
c&\ket{n_{r-1}, \dots , n_2, n_1, n_1 \pm 1,  0}\nonumber\\
&  +
d\ket{n_{r-1}, \dots , n_2, n_1, n_1 \pm 1,  1}\label{eq:tr2}
\end{align}
In analogy with the definition for Markov Chains, we call $r$ the 
\emph{order} of the quantum walk.

\section{Order 2 walks}

The state space is composed of the families of vectors
\begin{align}
\ket{n-1, n, 0},\qquad 
\ket{n-1, n, 1},\qquad 
\ket{n+1, n, 0},\qquad
\ket{n+1, n, 1}
\label{eq:states}
\end{align}
for $n\in Z$. In what follows, we will refer, for obvious reasons, to 
$\ket{n-1, n, p}$ as a right-mover, and to $\ket{n+1, n, p}$ as a left-mover.
Following \cite{ambainis}, it will suit us also to split the transitions (Eq. \ref{eq:tr1}, \ref{eq:tr2})
into two steps, a ``coin flip'' operator $C$ and a ``shift'' operator $S$:
\begin{align}
C:\qquad \ket{n_2, n_1, 0} &\longrightarrow a \ket{n_2, n_1, 0} + b \ket{n_2, n_1, 1} \\
C:\qquad \ket{n_2, n_1, 1} &\longrightarrow c \ket{n_2, n_1, 0} + d \ket{n_2, n_1, 1} \\
S:\qquad \ket{n_2, n_1, p} &\longrightarrow \ket{n_1, n_1\pm 1, p}
\end{align}
We investigate in what follows the possibilities for the shift operator $S$.
Suppose $S$ sends both $\ket{n-1, n, 0}$ and $\ket{n+1, n, 0}$ to the same 
vector, say $\ket{n, n+1, 0}$ (thus, in our parlance, for $p=0$, it sends both
left and right movers to right movers). One observes immediately that this is 
not really a 2nd.\ order chain (its behaviour does not depend on $n_2$, only on 
$n_1$). Indeed, on our state space $\ket{n_2, n_1, p}$ it is not even unitary
(even though it would be on the state space $\ket{n, p}$ of an order 1 walk).
For the behaviour with $p=1$, we have two possibilities:
\begin{enumerate}
\item $S$ sends both $\ket{n-1, n, 1}$ and $\ket{n+1, n, 1}$ in the same direction
(whether left or right). In this case, again $S$ behaves as a first order 
transition, and the whole analysis is that of a 1st.\ order quantum walk.
\item $S$ sends $\ket{n-1, n, 1}$ and $\ket{n+1, n, 1}$ in different directions. So,
for $p=1$, $S$ behaves like a 2nd.\ order chain. In this case, it turns out that
the combined behaviour does not give an invertible transition: i.e. the 
transition matrix is not unitary.
\end{enumerate}

Because of these arguments, to construct a bona fide 2nd.\ order walk, $S$ needs
to send $\ket{n-1, n, p}$ to a different state than it sends $\ket{n+1, n, p}$, 
for both values of $p$. The four possibilities are described in Table 
\ref{tab:shift}.
\begin{table}[h]
\begin{center}
\begin{tabular}{l||l|l|l|l|}
Initial State & \multicolumn{4}{c}{Final State}\\
\hline\hline
& Case a & Case b & Case c & Case d\\
\hline
$\ket{n-1,n,0}$ & $\ket{n,n+1,0}$ & $\ket{n,n+1,0}$ &
 $\ket{n,n-1,0}$ & $\ket{n,n-1,0}$\\ 
\hline
$\ket{n-1,n,1}$ & $\ket{n,n+1,1}$ & $\ket{n,n-1,1}$ &
 $\ket{n,n+1,1}$ & $\ket{n,n-1,1}$\\ 
\hline
$\ket{n+1,n,0}$ & $\ket{n,n-1,0}$ & $\ket{n,n-1,0}$ &
 $\ket{n,n+1,0}$ & $\ket{n,n+1,0}$ \\ 
\hline
$\ket{n+1,n,1}$ & $\ket{n,n-1,1}$ & $\ket{n,n+1,1}$ &
 $\ket{n,n-1,1}$ & $\ket{n,n+1,1}$ \\ 
\hline
\end{tabular}
\caption{Action of shift operator $S$\label{tab:shift}}
\end{center}
\end{table}
There is a simple way to view these cases, as follows.
 Depending on the value of 
the coin state $p$, one either transmits or reflects the walk:
\begin{description}
\item[Transmission] corresponds to 
$\ket{n-1,n,p} \longrightarrow \ket{n,n+1,p}$ and
$\ket{n+1,n,p} \longrightarrow \ket{n,n-1,p}$
(i.e.\ the particle keeps walking in the same direction it was going in)
\item[Reflection] corresponds to 
$\ket{n-1,n,p} \longrightarrow \ket{n,n-1,p}$ and
$\ket{n+1,n,p} \longrightarrow \ket{n,n+1,p}$
(i.e.\ the particle changes direction)
 \end{description}
We re-phrase in Table \ref{tab:shift1} the action of $S$ described in Table \ref{tab:shift}.
\begin{table}[h]
\begin{center}
\begin{tabular}{c||l|l|l|l|}
 Value of $p$ & \multicolumn{4}{c}{Action of $S$} \\ 
\hline\hline
& Case a & Case b & Case c & Case d\\
\hline
0 & Transmit & Transmit & Reflect & Reflect \\ 
\hline
1 & Transmit & Reflect & Transmit & Reflect\\
\hline
\end{tabular}
\caption{Action of shift operator $S$\label{tab:shift1}}
\end{center}
\end{table}

\subsection{Initial Conditions}
We must clarify how to initialize the walk, since at the very beginning, we 
cannot run a 2nd.\ order chain without any history. ``Starting'' at position 
-1, we then move to position 0 (which can be done using a 
first order quantum walk).
This creates the state $\ket{-1,0,0}$, and 
from there on we can run the second order 
operations described above.

\subsection{The Hadamard Walk}
We observe that Cases (a) and (d) do not lead to any interesting features. In 
Case (a), the particle just moves uniformly right or left, depending on the 
initial state. If the initial state is a superposition of left- and right- 
movers, the walk progresses simultaneously right and left. For Case (d), the 
walk ``stays put'', oscillating forever between $n$ and $n+1$ for some value 
of $n$. In both these cases in fact, the coin flip operator $C$ plays no role
(since the action of $S$ is independent of $p$), so there is nothing quantum
about these walks.

However, cases (b) and (c) do yield results of interest. To 
analyze these, we choose a particular coin flip operator $C$ corresponding 
to the Hadamard walk:
\begin{description}
\item[Classically] $C$ sends $\ket{n_2,n_1,p}$ to either 
$\ket{n_2,n_1,0}$ or $\ket{n_2,n_1,1}$ with equal probability 1/2
(fair coin toss).
\item[Quantumly]
\begin{align}
C:\qquad \ket{n_2, n_1, 0} &\longrightarrow 
\dfrac{1}{\sqrt{2}} (\ket{n_2, n_1, 0} +  \ket{n_2, n_1, 1}) 
\label{eq:had1}\\
C:\qquad \ket{n_2, n_1, 1} &\longrightarrow 
\dfrac{1}{\sqrt{2}}( \ket{n_2, n_1, 0} -  \ket{n_2, n_1, 1} )
\label{eq:had2}
\end{align}
 \end{description}
The equations \ref{eq:had1} and \ref{eq:had2} correspond to 
$a = b = c = -d = 1/\sqrt{2} $ which is known as the Hadamard walk.

For both cases (b) and (c) it should be clear that in the classical 
case, we end up with the standard (classical) random walk: In each 
case, transmission and reflection just correspond to picking one of 
two different choices (right or left) at each step.

Let us consider case (c): The first few steps of a standard 
quantum (Hadamard) walk starting at position $n$ would be
\begin{align}
\ket{n, 0} &\longrightarrow 
\dfrac{1}{\sqrt{2}} (\ket{n-1, 0} +  \ket{n+1, 1}) \longrightarrow\\
& \dfrac{1}{2} ( \ket{n-2, 0} + \ket{n,1} + \ket{n,0} - \ket{n+2,1} )
\longrightarrow\\
&\dfrac{1}{2\sqrt{2}}( \ket{n-3, 0} + \ket{n-1,1}
+ \ket{n-1,0}  - \ket{n+1,1}\nonumber\\
&\qquad + \ket{n-1,0} + \ket{n+1,1} - \ket{n+1,0} + \ket{n+3,1}).
\label{eq:wlk1}
\end{align}
Thus after the third step of the walk we see destructive interference 
(cancellation of 4th.\ and 6th.\ terms in expression \ref{eq:wlk1}) and 
constructive interference (addition of 3rd.\ and 5th.\ terms in 
expression \ref{eq:wlk1}).
However for case (c) the first few steps are for example
\begin{align}
&\ket{n-1, n, 0} \longrightarrow 
\dfrac{1}{\sqrt{2}} (\ket{n, n-1, 0} +  \ket{n, n+1, 1}) \longrightarrow\\
& \dfrac{1}{2} ( \ket{n-1, n, 0} + \ket{n-1, n-2, 1} 
+ \ket{n+1, n,0} - \ket{n+1, n+2,1} )
\longrightarrow\\
&\dfrac{1}{2\sqrt{2}}( \ket{n,n-1, 0} + \ket{n,n+1,1}
+ \ket{n-2,n-1,0}  - \ket{n-2,n-3,1}\nonumber\\
&\qquad + \ket{n,n+1,0} + \ket{n,n-1,1} 
- \ket{n+2,n+1,0} + \ket{n+2,n+3,1}) \longrightarrow\\
& \dfrac{1}{4} ( \ket{n-1, n, 0} + \ket{n-1, n-2, 1}
+ \ket{n+1, n,0} - \ket{n+1, n+2,1}\nonumber\\
&\quad + \ket{n-1,n-2,0} +\ket{n-1,n,1}
-\ket{n-3,n-2,0} +\ket{n-3,n-4,1}\nonumber\\
&\quad +\ket{n+1,n,0} +\ket{n+1,n+2,1}
+\ket{n-1,n,0} -\ket{n-1,n-2,1}\nonumber\\
&\quad -\ket{n+1,n+2,0} - \ket{n+1,n,1}
+\ket{n+3,n+2,0} -\ket{n+3,n+4,1}).\label{eq:wlk2}
\end{align}
After three steps, there is no interference (constructive or destructive), but 
the interference appears after step four (e.g., in expression \ref{eq:wlk2},
we can cancel the 2nd.\ and 12th.\ terms, and we can add term 3 and term 9, 
etc.). Thus we can see this walk differs both from the classical random 
walk and from the standard (Hadamard) quantum walk.

\section{Amplitudes}
(This section follows closely the approach taken in Appendix A of \cite{coins}).
We now derive analytical expressions for the wavefunction amplitudes in case 
(c) of table \ref{tab:shift1}, using as quantum coin flip the Hadamard 
transition
\ref{eq:had1} and \ref{eq:had2}.

For the 1-dimensional walk, we view the progression as a sequence of left
($L$) and right ($R$) moves. In general there are many paths to reach a 
particular final state: We need to sum over the amplitudes of these 
different paths (with appropriate phases) to obtain the amplitude for 
that final state.

As a quick example, for the classical case, what is the probability of 
ending at position 1 in a 3-step walk that starts at the origin? The possible
walks ending at 1 are $LRR$, $RLR$ or $RRL$. The total number of possible 
3-step walks is $2^3=8$. So the probability of finishing at positions 1 is
3/8.

In the notation of expression \ref{eq:states}, let our initial state be $\ket{-1,0,0}$
(so the walk starts at the origin) and let us take $n$ steps in the walk.
It should first of all be obvoius that as in the classical case, if $n$ is
odd/even, we can only finish up at an odd/even integer position
(respectively) on the 1-dim lattice.
 Let $N_L$ 
be the number of left moves, and $N_R$ the number of right moves.

\begin{lemma}
\label{lem:phase} We refer to an `isolated' $L$ (respectively $R$) as
one which is not bordered on either side by another $L$ (respectively $R$).
Let $N_L^1$ (respectively $N_R^1$) be the number of isolated $L$s
(respectively isolated $R$s) in the sequence of steps of the walk. Then,
the quantum phase associated with this sequence is 
\begin{align}
(-1)^{N_L+N_R+N_L^1+N_R^1}
\label{eq:ph1}
\end{align}
\end{lemma}
\begin{proof}
In what follows, we first analyze the sequence of $L$s (identical
arguments will apply to the $R$s). An isolated $L$ does not contribute
to the phase, nor does the pair $LL$ bordered by $R$s. The first 
sequence of $L$s that can contribute is $LLL$: In our previous language,
this corresponds to \emph{transmit} followed by \emph{transmit}. After the 
first $L$, the coin state is 0, after the second it is 1, and after the 
third it is 1. It is the transition from 1 to 1 in the coin state that 
gives the factor of $-1$ from the Hadamard walk.

So in general, a sequence of $j$ $LL\dots L$s will give a phase 
contribution of $(-1)^{j}$ for $j>2$.

Now examine clusters of $L$s of size greater than 2. If we have 2 such
clusters, we can move one $L$ from the first cluster to the second, without
changing the overall phase contribution. In such a move, the contribution
of the 1st.\ cluster decreases by a factor of -1, while that of the 2nd.\
increases by the same factor. Suppose we repeat this process, to shrink all
but one of the large clusters to clusters of size 2. We end up with a sequence
that looks like

\begin{align}
\dots RLR\dots RLR\dots RLLR\dots RLR\dots RLLR\dots 
R\underbrace{LLLLL\dots L}_{\text{One large cluster of $L$s}}R\dots
\label{eq:ph2}
\end{align}
Denote by $C_L$ the total number of $L$ clusters. Then the total number of 
$L$ clusters of size 2 is $C_L-N_L^1-1$. So, the size of the one large 
cluster of $L$s is $N_L-N_L^1 - 2(C_L-N_L^1-1) = N_L+N_L^1-2C_L+2$. Its 
phase contribution is therefore $(-1)^{N_L+N_L^1}$.

Since analogous arguments apply for sequences of $R$s, the total phase 
contribution is $(-1)^{N_L+N_R+N_L^1+N_R^1}$.
\qed
\end{proof}

After an $n$- step walk, we want to know what is the probability the particle
is in position $k$. From previous arguments, $(-1)^n = (-1)^k$ and 
$-n \leq k \leq n$. Four possible final quantum states correspond in our 
model to the particle terminating at $k$:

\begin{align}
\underbrace{\ket{k-1, k, 0}}_{\parbox{2cm}{\small sequence\\ ending $\dots LR$}},\qquad 
\underbrace{\ket{k-1, k, 1}}_{\parbox{2cm}{\small sequence\\ ending $\dots RR$}},\qquad 
\underbrace{\ket{k+1, k, 0}}_{\parbox{2cm}{\small sequence\\ ending $\dots RL$}},\qquad
\underbrace{\ket{k+1, k, 1}}_{\parbox{2cm}{\small sequence\\ ending $\dots LL$}}
\label{eq:final}
\end{align}
Let us denote by $a_{kLR}, a_{kRR}, a_{kRL}, a_{kLL}$ the amplitudes of these 4 
states in the final wavefunction $\Psi$. Then the probability when we measure of 
finding the particle at position $k$ is 
\begin{align}
\scalarnorm{a_{kLR}}^2 +
\scalarnorm{a_{kRR}}^2 +
\scalarnorm{a_{kRL}}^2 +
\scalarnorm{a_{kLL}}^2
\label{eq:prob1}
\end{align}
Before calculating the amplitudes, we need another technical lemma.

\begin{lemma}
\label{lem:ones} Consider a composition (ordered partition) of the integer
$n$ into $C$ parts, and let $N^1$ be the number of 1s in the composition. Then 
either
\begin{enumerate}
\item $n=C=N^1$\\
\mbox{or}
\item $\max (0, 2C-n)\leq N^1\leq C-1$.
\end{enumerate}
\end{lemma}
\begin{proof}
Case 1.\ is trivial: It is the composition of $n$ into $N^1$ 1s.
For case 2., the upper limit is also trivial: The largest number of individual 1s 
 we can get is $C-1$, which is the composition 
\begin{align}
n=\underbrace{1+1+1+\dots +1}_{(C-1) \text{terms}}+(n-(C-1))
\end{align}
For the lower limit, assume $C<n/2$. It is always possible to 
write down a composition with few terms, without using any 1s. Specifically, 
we can write the first $C-1$ terms as 2, and the last term as the remainder
($n- 2(C-1)$), which is greater than 2 by assumption.

Now assume $C\geq n/2$. The least number of 1s in the composition is obtained 
by writing as many 2s as possible. Suppose we have $r$ 2s, and the other terms 
are 1. Then $2r+N^1=n$. Since $r=C-N^1$, we have that $N^1=2C-n$, and the result
follows.
\qed
\end{proof}
We define the combinatorial symbol
\begin{align}
\left(
\begin{array}{c}
a\\
b\\
c
\end{array}
\right)
= 
\left(
\begin{array}{c}
a\\
b
\end{array}
\right)
\left(
\begin{array}{c}
c-a-1\\
a-b-1
\end{array}
\right)
=\frac{a!}{b!(a-b)!}
\frac{(c-a-1)!}{(c-2a+b)!(a-b-1)!}
\end{align}

\begin{theorem}\label{thm:main}
The amplitudes $a_{kLL}, a_{kLR}, a_{kRL}, a_{kRR}$ for the final 
states given in Equation \ref{eq:final} are
\begin{align}
 2^\frac{n}{2}a_{kLL}&=
\sum_{C=2}^{N_L-1}\quad
\sum_{\parbox[c]{1.5cm}{\tiny $N_L^1=\max (1,$\\ $2C-N_L)$}}^{C-1}\quad
\sum_{\parbox[c]{1.5cm}{\tiny $N_R^1=\max (0,$\\ $2C-N_R-2)$}}^{C-2} 
(-1)^{n+N_L^1+N_R^1}\nonumber\\
&
\frac{N_L^1(C-N_L^1)}{C(C-1)}
\combn{C}{N_L^1}{N_L}
\combn{C-1}{N_R^1}{N_R}
\nonumber\\
&
+\sum_{\parbox[c]{1.5cm}{\tiny $N_L^1=\max(1,$\\ $2N_R-N_L+2)$}}^{N_R}
 (-1)^{N_L+N_L^1}
\frac{N_L^1(N_R-N_L^1+1)}{N_R(N_R+1)}\combn{N_R+1}{N_L^1}{N_L}
\label{eq:aLL} 
\end{align}
\begin{align}
 2^\frac{n}{2}a_{kLR}&=
\sum_{C=2}^{N_L-1}\quad
\sum_{\parbox[c]{1.5cm}{\tiny $N_L^1=\max (1,$\\ $2C-N_L)$}}^{C-1}\quad
\sum_{\parbox[c]{1.5cm}{\tiny $N_R^1=\max (1,$\\ $2C-N_R)$}}^{C-1} 
(-1)^{n+N_L^1+N_R^1}
\nonumber
\\
&
\frac{N_L^1(N_R^1)}{C^2}
\combn{C}{N_L^1}{N_L}
\combn{C}{N_R^1}{N_R}
+\delta_{N_L,N_R}
\nonumber
\\
&
+\sum_{\parbox[c]{1.5cm}{\tiny $N_L^1=\max(1,$\\ $2N_R-N_L)$}}^{N_R-1}
 (-1)^{N_L+N_L^1}
\frac{N_L^1}{N_R}\combn{N_R}{N_L^1}{N_L}
+\sum_{\parbox[c]{1.5cm}{\tiny $N_R^1=\max(1,$\\ $2N_L-N_R)$}}^{N_L-1}
 (-1)^{N_R+N_R^1}
\frac{N_R^1}{N_L}\combn{N_L}{N_R^1}{N_R}
\label{eq:aLR} 
\end{align}
\begin{align}
 2^\frac{n}{2}a_{kRL}&=
\sum_{C=2}^{N_L-1}\quad
\sum_{\parbox[c]{1.5cm}{\tiny $N_L^1=\max (2,$\\ $2C-N_L)$}}^{C-1}\quad
\sum_{\parbox[c]{1.5cm}{\tiny $N_R^1=\max (0,$\\ $2C-N_R-2)$}}^{C-2} 
(-1)^{n+N_L^1+N_R^1}
\nonumber
\\
&
\frac{N_L^1(N_L^1-1)}{C(C-1)}
\combn{C}{N_L^1}{N_L}
\combn{C-1}{N_R^1}{N_R}
 + \delta_{N_L-1,N_R}
\nonumber
\\
&
+\sum_{\parbox[c]{1.5cm}{\tiny $N_L^1=\max(2,$\\ $2N_R-N_L+2)$}}^{N_R}
 (-1)^{N_L+N_L^1}
\frac{N_L^1(N_L^1-1)}{N_R(N_R+1)}\combn{N_R+1}{N_L^1}{N_L}
\nonumber
\\
&
+\sum_{\parbox[c]{1.5cm}{\tiny $N_R^1=\max(0,$\\ $2N_L-N_R-2)$}}^{N_L-2}
 (-1)^{N_R+N_R^1}
\combn{N_L-1}{N_R^1}{N_R}
\label{eq:aRL} 
\end{align}
\begin{align}
2^\frac{n}{2}a_{kRR}&= 
\sum_{C=1}^{N_L-1}\quad
\sum_{\parbox[c]{1.5cm}{\tiny $N_L^1=\max (1,$\\ $2C-N_L)$}}^{C-1}\quad
\sum_{\parbox[c]{1.5cm}{\tiny $N_R^1=\max (0,$\\ $2C-N_R)$}}^{C-1} 
(-1)^{n+N_L^1+N_R^1}
\nonumber
\\
&
\frac{N_L^1(C-N_R^1)}{C^2}
\combn{C}{N_L^1}{N_L}
\combn{C-1}{N_R^1}{N_R}
\nonumber
\\
&
+\sum_{\parbox[c]{1.5cm}{\tiny $N_R^1=\max(0,$\\ $2N_L-N_R)$}}^{N_L-1}
 (-1)^{N_R+N_R^1}
\frac{N_L-N_R^1}{N_L}\combn{N_L}{N_R^1}{N_R}
\label{eq:aRR} 
\end{align}
where $k=N_R-N_L,\ n=N_R+N_L-2$ and $\delta$ is the standard 
kronecker delta function ($\delta_{p,q} = 1$ if $p=q$ and zero otherwise).

\end{theorem}
\begin{proof}
Because of its slightly lengthy and technical nature, we relegate the
proof to Appendix A.
\qed
\end{proof}
\section{Simulations and Analysis}
\begin{figure}[!h]
\begin{minipage}{0.46\linewidth}
\hspace*{-1cm}
\includegraphics*[width=8cm, height=7cm, viewport = 40 250 600 770]{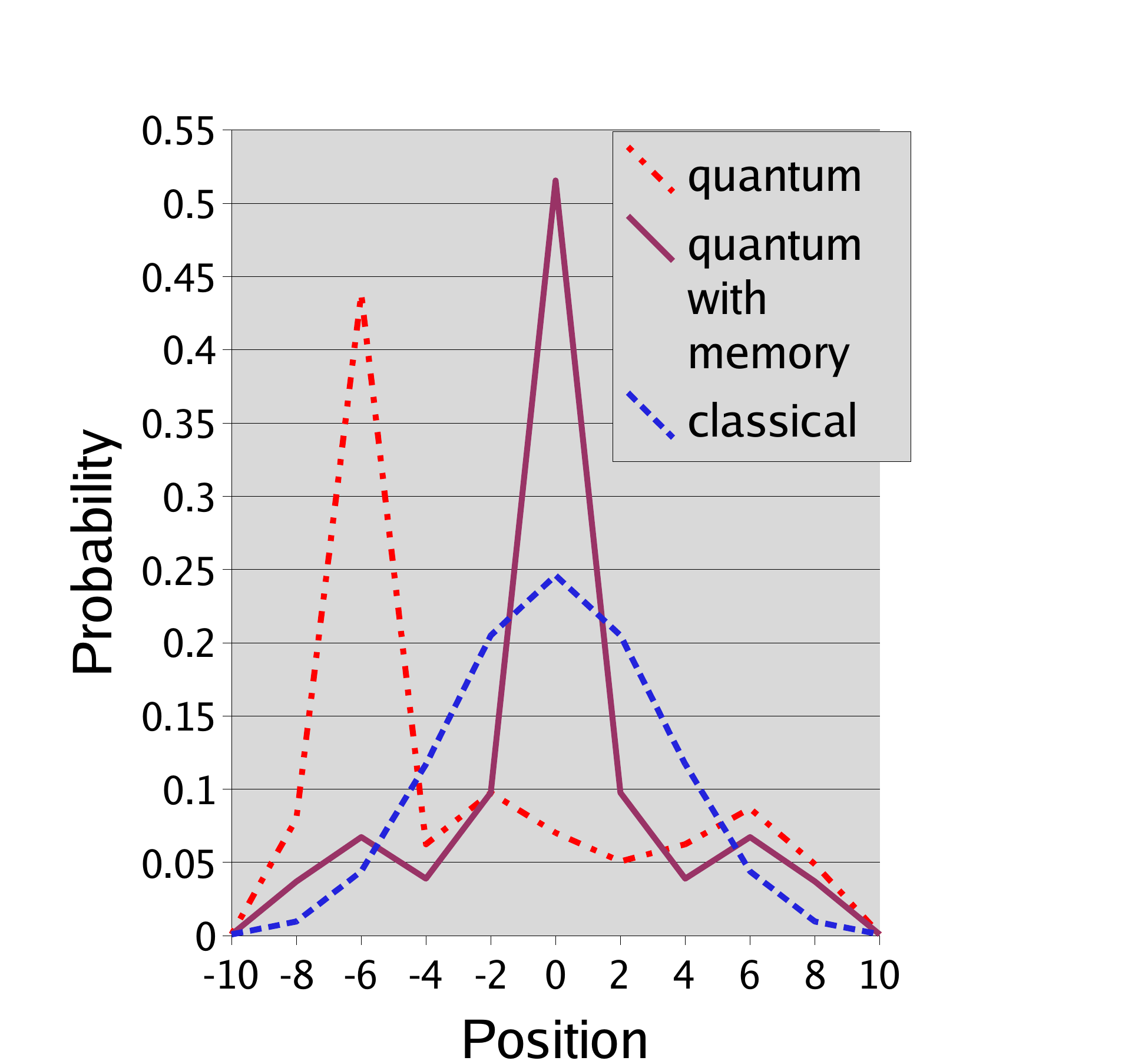}
\caption{Probability Distribution after 10 steps}
\label{fig:10step}
\end{minipage}%
\hfill
\begin{minipage}{0.46\linewidth}
\hspace*{-0.5cm}
\includegraphics*[width=6cm, height=7cm, viewport = 10 50 600 750]{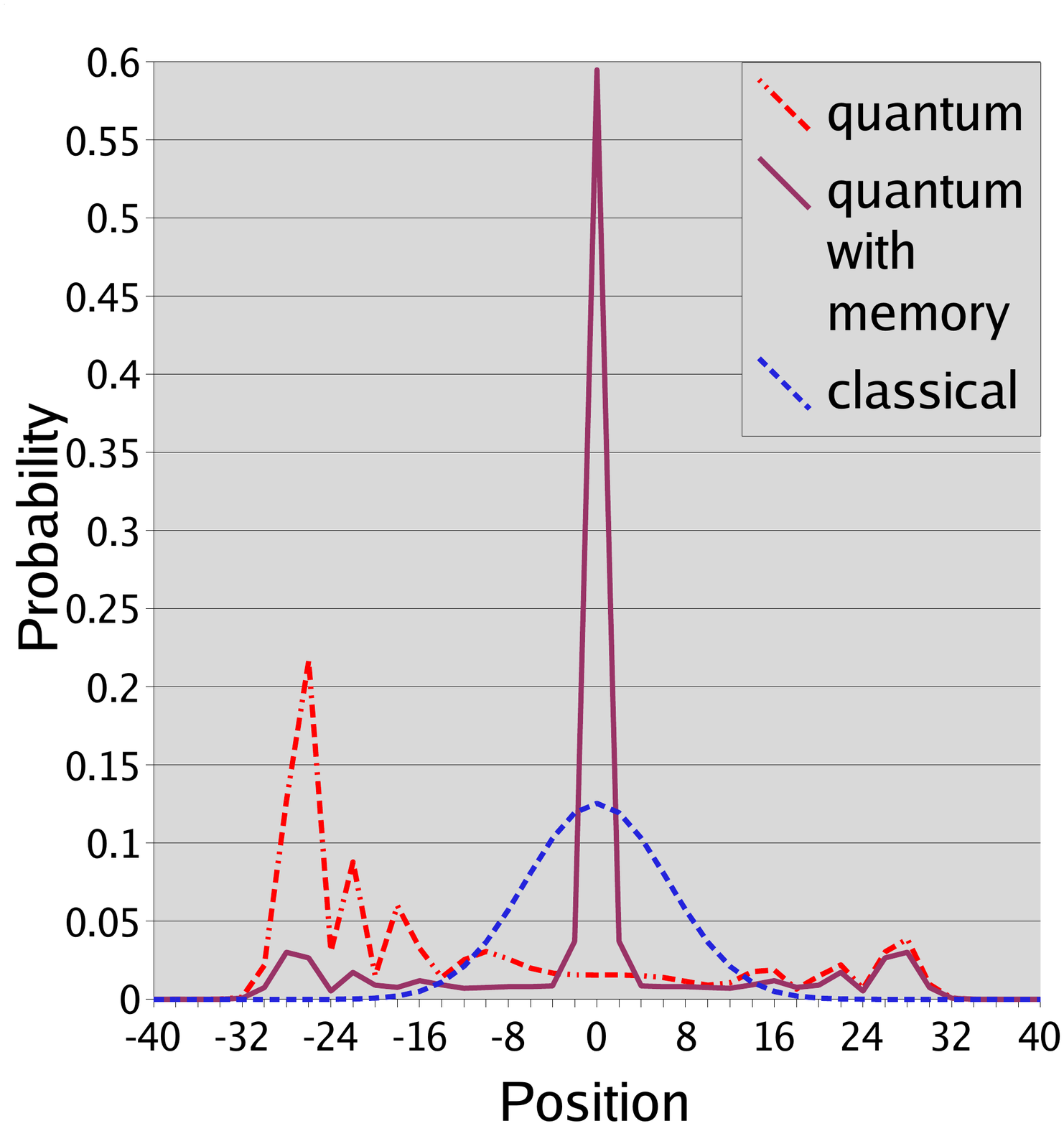}
\caption{Probability Distribution after 40 steps}
\label{fig:100step}
\end{minipage}
\end{figure}

We show in Figures \ref{fig:10step} and \ref{fig:100step} the amplitudes for 
the 3 different kinds of walks (classical, quantum, quantum with memory). The 
simulations are carried out using AXIOM \cite{axiom}. For completeness, we 
include in Appendix B the commented code for generating the Quantum Walk with
memory.

In figures \ref{fig:10step} and \ref{fig:100step}, the initial states for 
the three cases are $\ket{0}$, $\ket{0, 0}$ and $\ket{-1,0,0}$ (by abuse of 
notation, the ket vector here $\ket{0}$ represents the classical case).
As has been pointed out by a number of authors (see e.g. \cite[Appendix A]{nayak})
in the quantum case we can choose a more symmetric initial state (still of 
course representing the particle starting at the origin). In general this will
give rise to a different probability distribution. For the quantum walk we 
start at $(\ket{0,0} + \ket{0,1})/\sqrt{2}$ and for our walk with memory, we start
at $(\ket{-1,0,0} + \ket{-1,0,1} + \ket{1,0,0} + \ket{1,0,1})/2$. The probability
distributions for these cases (for a 40-step walk) are plotted in figure
\ref{fig:SYMstep}.
\begin{figure}[!h]
\centering
\hspace*{-1cm}
\includegraphics*[width=10cm, height=7cm, viewport = 40 250 600 770]{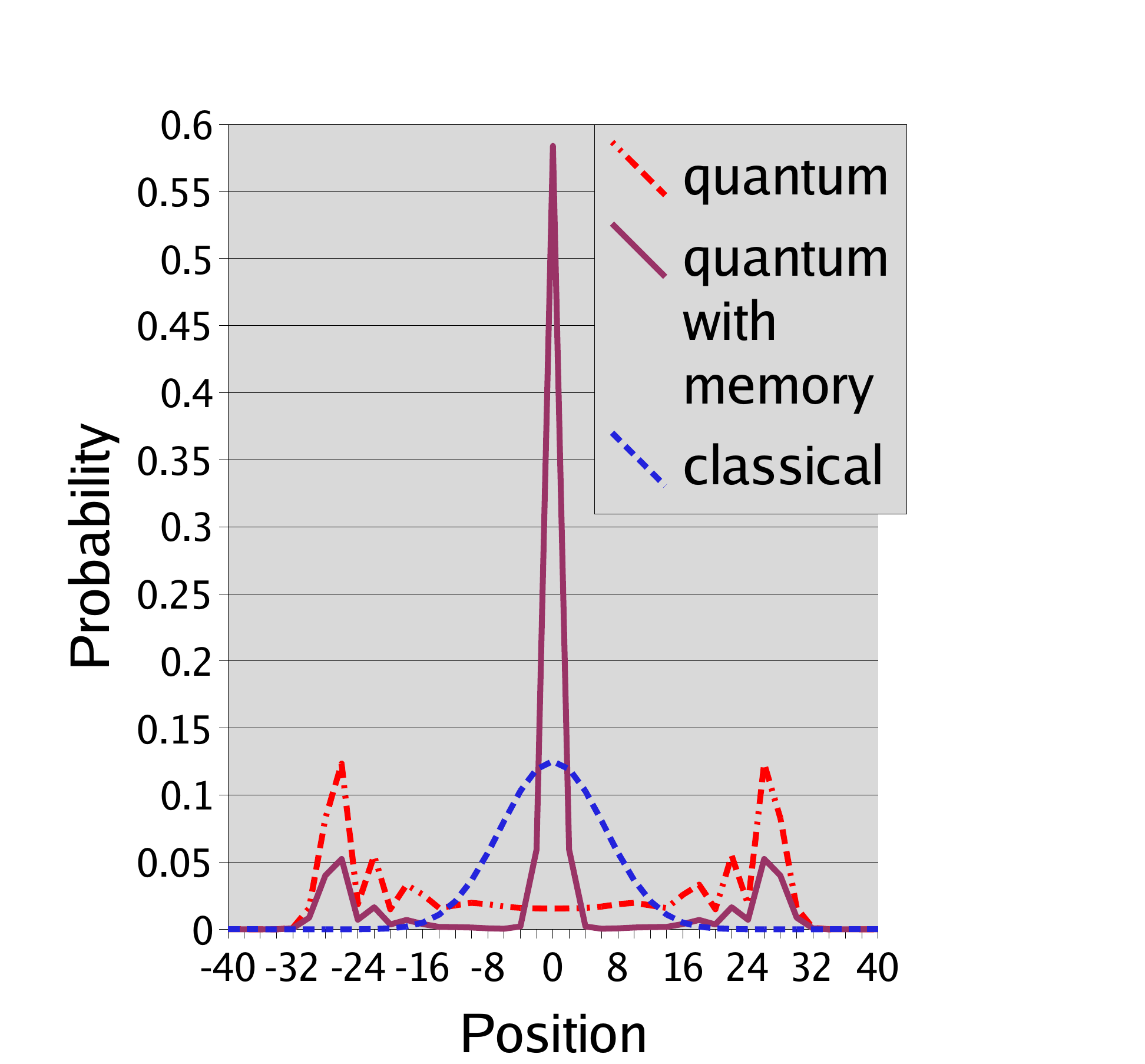}
\caption{Probability Distribution after 40 steps for symmetric initial state}
\label{fig:SYMstep}
\end{figure}

What is immediately noticeable is the high probability that the quantum walk 
with memory stays at the origin (even after 40 steps, it has more than 50\% 
chance of being found at the origin). In the terminology of 
Konno \cite{konno1}, we say the particle 
is \textbf{localized} at the
origin. Also of note are the smaller peaks that 
occur quite a distance from the origin (at $\pm 6$ in Figure \ref{fig:10step}
and at $\pm 28$ in Figure \ref{fig:100step}). The distribution is symmetric about
zero, except for the one specific case $N_R = N_L\pm 1$ (i.e. in a walk 
with an odd number of steps, the probability of finding the particle at positions
$\pm 1$ is not the same).

\begin{claim}
As the quantum walk with memory becomes infinitely long, for even 
$n$ there is still a chance
of over 50\% of finding the particle at the origin!
\end{claim}
\begin{proof}
The proof proceeds by setting $n=2j$ and using an inductive argument on $j$
(the particle can only be at the origin for an even number of steps). Let us 
denote by $a_{k\ast\ast}(n)$ the dependence of the amplitude on the 
number of steps $n$, where $\ast\ast$ is one of $LL, LR, RL, RR$. The argument
focuses on the dependence of $S= \{a_{0LR}(n), a_{0RL}(n) \} $ on their 
equivalents two steps earlier $S^\dagger = \{a_{0LR}(n-2), a_{0RL}(n-2) \} $.
\begin{description}
\item[Base Case]
For $n=2$, $a_{0LR}(2) = a_{0RL}(2) = 1/2$ are 
the only terms contributing to the probability at the origin.
\item[Inductive Step]
Let us consider $a_{0LR}(n)$ and $a_{0RL}(n)$. Assume the amplitudes 
$a_{0LR}(n-2)$ and $a_{0RL}(n-2)$ are both positive and sum to 1 (as in 
the base case).
\begin{description}
\item[Amplitude of $\ket{-1,0,0}$] This corresponds to $a_{0LR}(n)$. There are 
2 contributions from $a_{0\ast\ast}(n-2)$:
\begin{description}
\item[Contribution from $a_{0LR}(n-2)$] The particle moves left and then right.
The phase contribution stays positive. The amplitude factor is $(1/\sqrt{2})^2=0.5$.
\item[Contribution from $a_{0RL}(n-2)$] The particle moves left and then right.
The phase contribution stays positive. The amplitude factor is $(1/\sqrt{2})^2=0.5$.
 \end{description}
Thus the total amplitude contribution is 
$0.5 a_{0LR}(n-2) + 0.5 a_{0RL}(n-2) = 0.5 a_{0LR}(n-2) + 0.5 (1-a_{0LR}(n-2)) = 0.5$
\item[Amplitude of $\ket{1,0,0}$] This corresponds to $a_{0RL}(n)$. There are 
2 contributions from $a_{0\ast\ast}(n-2)$:
\begin{description}
\item[Contribution from $a_{0LR}(n-2)$] The particle moves right and then left.
The phase contribution stays positive. The amplitude factor is $(1/\sqrt{2})^2=0.5$.
\item[Contribution from $a_{0RL}(n-2)$] The particle moves right and then left.
The phase contribution stays positive. The amplitude factor is $(1/\sqrt{2})^2=0.5$.
 \end{description}
Thus the total amplitude contribution is 
$0.5 a_{0LR}(n-2) + 0.5 a_{0RL}(n-2) = 0.5 a_{0LR}(n-2) + 0.5 (1-a_{0LR}(n-2)) = 0.5$
 \end{description}
Thus we have constructive interference for both amplitudes in the transitions 
from set $S$ to set $S^\dagger$.
 \end{description}
We need to show further that amplitudes $a_{0LL}(n-2)$ and $a_{0RR}(n-2)$ will 
not decrease our amplitudes for $a_{0\ast\ast}(n)$. Let us consider 
$a_{0LL}(n-2)$. Again, the two contributions arise from moving either $RL$ or
$LR$. The amplitude factor, as before, is 0.5. But the phase factor 
contributions are opposite: For $RL$ it is positive, while for $LR$ it is 
negative. This adds $0.5a_{0LL}(n-2)$ to the amplitude $a_{0RL}(n)$ and 
subtracts $0.5a_{0LL}(n-2)$ from the amplitude $a_{0LR}(n)$. Letting
$\epsilon = 0.5a_{0LL}(n-2)$, since $a_{0LR}(n-2)$ and $a_{0RL}(n-2)$
are two positive numbers summing to 1, so are 
$a_{0LR}(n-2) - \epsilon$ and $a_{0RL}(n-2) + \epsilon$.

A similar argument holds for the contribution from $a_{0RR}(n-2)$. We have shown,
for all even $n$, that $a_{0LR}(n)$ and $a_{0RL}(n)$ are two positive numbers 
summing to one, and hence their contribution to the probability 
$|a_{0LR}(n)|^2 + |a_{0RL}(n)|^2$ is at least 0.5. Note that in general,
$a_{0LL}(n)$ and $a_{0RR}(n)$ will be non-zero, and will also contribute 
to the probability at zero (though it turns out this contribution is small).
\qed
\end{proof}
\section{Conclusion}
We have defined a new kind of quantum walk with (two-step) memory, and investigated 
its properties. We see it exhibits some similarities with the classical random walk
(symmetric probability distribution, high probability at the origin), and other 
similarities with the quantum (Hadamard) walk (oscillatory behaviour, ``tails'' that
propagate faster than in the classical case). We prove the remarkable feature of localization at the origin:
in the $n\to \infty$ limit, for a symmetric initial state,
 the probability the particle is found at the origin 
is not less than 0.5.

A referee has pointed out the work of Kendon (\cite{kendon}) on decoherence in 
quantum walks, where probability distributions that peak at the origin are 
also obtained. However, a fundamental difference is that our peak at the 
origin is independent of the walk length, unlike the results for the decoherence
case. It is worth examining this in more detail to see if there are other 
similarities in the results. 

Other models of quantum walks with
history have been constructed (see \cite{stang,flitney}) by using 
multiple coins or modifying the Hamiltonian. We find intriguing that
the probability distribution for the 2-coin model in 
\cite[Figure 4]{flitney} seems close in shape to our
results (in e.g. Figure \ref{fig:100step}), with again the fundamental 
difference that in our case the peak at the origin is much larger 
and independent of the walk length. Models with multiple internal
states (\cite{inui, segawa}) have also been found to exhibit 
memory effects and localization.
\section*{Acknowledgements}
This work was carried out while the author was visiting the Laboratoire de 
Recherche en Informatique (Universit\'e Paris Sud). We thank Miklos Santha
and other colleagues at LRI for their hospitality. We also thank the anonymous 
referees for their comments and corrections.
\appendix
\section{Appendix: Proof of Theorem \ref{thm:main}}
Here we prove Theorem \ref{thm:main}.
We denote by $C_L$ (respectively $C_R$) the number of clusters of $L$s
(respectively $R$s) in a sequence of $L$s and $R$s representing a particular
walk. For example, in ${\color{Corcra}L}R{\color{Corcra}LL}R{\color{Corcra}LLL}RR{\color{Corcra}L}$, $C_L=4$ and $C_R=3$.

We examine firstly compositions of the integer $N_L$ into $C_L$ parts. Because 
of the phase dependence given in Lemma \ref{lem:phase}, we need to 
know how this composition depends on $N_L^1$, the number of clusters of size
one. The number of distinct compositions of $N_L$ with $C_L$ parts, and with no
part of size 1 is
\begin{align}
\left(
\begin{array}{c}
N_L-C_L-1\\
C_L-1
\end{array}
\right)
\end{align}
(see for example \cite[page 15]{merlini} ). If we want exactly one part of 
size 1, we take a composition of $N_L-1$ into $C_L-1$ parts, no part 
of size 1, and add the one 
cluster of size 1. The number of ways we can do this is
\begin{align}
N_L
\left(
\begin{array}{c}
N_L-C_L-1\\
C_L-2
\end{array}
\right).
\end{align}
In the general case, we want to add $N_L^1$ clusters of size one to 
a composition of $N_L-N_L^1$ into $C_L-N_L^1$ parts, none of which is one.
We can imagine having $C_L$ boxes: $N_L^1$ of them will be filled by 
clusters of size one (in $C_L!/(C_L-N_L^1)!N_L^1!$ distinct ways), and the 
remaining $C_L-N_L^1$ boxes will take a composition of $N_L-N_L^1$ into
$C_L-N_L^1$ parts, without any ones. So we get
\begin{align}
\parbox{3.6cm}{\small Number of compositions of $N_L$ into 
$C_L$ parts with exactly $N_L^1$ ones}
=
\left(
\begin{array}{c}
C_L\\
N_L^1
\end{array}
\right)
\left(
\begin{array}{c}
N_L-C_L-1\\
C_L-N_L^1-1
\end{array}
\right)
\defeq
\left(
\begin{array}{c}
C_L\\
N_L^1\\
N_L
\end{array}
\right).
\end{align}
(\emph{N.B.\ } This formula does not apply in the extreme case 
$N_L=C_L=N_L^1$, in which case the number of such compositions is
just 1.)
For fixed values of $N_L$ and $N_R$ (so a fixed final position $k$), the number 
of walks with $C_L$ left clusters ($N_L^1$ of size 1) 
 and $C_R$ right clusters ($N_R^1$ of size 1) is
\begin{align}
\left(
\begin{array}{c}
C_L\\
N_L^1\\
N_L
\end{array}
\right)
\left(
\begin{array}{c}
C_R\\
N_R^1\\
N_R
\end{array}
\right).
\end{align}
(Of course, $C_L$ and $C_R$ are not independent - they differ by at most 1.)
Putting in the phase factor from Lemma \ref{lem:phase}, the $\sqrt{2}$ factors
from Equations \ref{eq:had1} \ref{eq:had2}, and summing over
 $C_L, C_R, N_L^1, N_R^1$ we get the amplitude expression
\begin{align}
\sum_{C_L}
\sum_{C_R}
\sum_{N_L^1}
\sum_{N_R^1}
&
\frac{(-1)^{n+N_L^1+N_R^1}}{(2)^{n/2}}
\left(
\begin{array}{c}
C_L\\
N_L^1\\
N_L
\end{array}
\right)
\left(
\begin{array}{c}
C_R\\
N_R^1\\
N_R
\end{array}
\right).\label{eq:genamp1} 
\end{align}
We now derive the specific expressions for the four possible 
states with final position $k$. In all cases we take as initial 
state $\ket{-1,0,0}$ corresponding to a walk beginning 
$LR\dots $ (This means in particular that $N_L^1$ is at least 1,
corresponding to the first $L$.)
\begin{description}
\item[Final State $\ket{k+1,k,1}$]
This corresponds to a walk of the form $LR\dots LL$. Since the 
sequence begins and ends with an $L$, we have $C_R=C_L-1$, which 
removes the summation over $C_R$.
In general, $N_L\ge C_L\ge N_L^1$, and either all
three of these numbers are different or are equal.
In this particular case, the possible values of $C_L$ run
from 2 to $N_L-1$.

Let us examine $N_L^1$ and $N_R^1$. For a particular value of
$C_L$, $N_L^1$ will run from 1 to $C_L-1$ (the upper limit is not 
$C_L$ because of our observation that the 3 numbers $N_L$, $C_L$ 
and $N_L^1$ are either identical or all different from one another).
$N_R^1$ will run from 0 to $C_R$, i.e.\ from 0 to $C_L-1$.

Since the expression (\ref{eq:genamp1}) is for all walks, we need to 
restrict this to walks beginning with $LR$ and ending with $LL$.
For the composition of $N_L$ $L$s into $C_L$ clusters,
this forces the first cluster to be of size 1, and the last 
\emph{not} to be of size 1. The fraction of walks whose first 
cluster is of size one is $N_L^1/C_L$. Of these, the fraction that
do not have a cluster of size 1 at the end is $(C_L-N_L^1)/(C_L-1)$.
Putting all of this together, the sum 
(\ref{eq:genamp1}) becomes
\begin{align}
\sum_{C_L=2}^{N_L-1}
\sum_{N_L^1=1}^{C_L-1}
\sum_{N_R^1=0}^{C_L-1}
&
\frac{(-1)^{n+N_L^1+N_R^1}}{(2)^{n/2}}
\frac{N_L^1(C_L-N_L^1)}{C_L(C_L-1)}
\left(
\begin{array}{c}
C_L\\
N_L^1\\
N_L
\end{array}
\right)
\left(
\begin{array}{c}
C_L-1\\
N_R^1\\
N_R
\end{array}
\right).\label{eq:amp1} 
\end{align}
We separate out from the sum the limiting case:
\begin{enumerate}
\item 
\underline{$N_R=C_R=N_R^1$} :
Using the result of Lemma \ref{lem:ones}, expression (\ref{eq:amp1}) becomes
\begin{align}
\sum_{\parbox[c]{1.5cm}{\tiny $N_L^1=\max (1,$\\ $2N_R-N_L+2)$}}^{N_R}
\frac{(-1)^{N_L+N_L^1}}{(2)^{n/2}}
\frac{N_L^1(N_R-N_L^1+1)}{N_R(N_R+1)}
\left(
\begin{array}{c}
N_R+1\\
N_L^1\\
N_L
\end{array}
\right).\label{eq:amp2} 
\end{align}
\item
\underline{$N_R > C_R > N_R^1$} : Using Lemma \ref{lem:ones} the amplitude is
\begin{align}
\sum_{C_L=2}^{N_L-1}
&
\sum_{\parbox[c]{1.5cm}{\tiny $N_L^1=\max (1,$\\ $2C_L-N_L)$}}^{C_L-1}
\sum_{\parbox[c]{1.5cm}{\tiny $N_R^1=\max (0,$\\ $2C_L-N_R-2)$}}^{C_L-2}
\frac{(-1)^{n+N_L^1+N_R^1}}{(2)^{n/2}}
\nonumber
\\
&
\frac{N_L^1(C_L-N_L^1)}{C_L(C_L-1)}
\left(
\begin{array}{c}
C_L\\
N_L^1\\
N_L
\end{array}
\right)
\left(
\begin{array}{c}
C_L-1\\
N_R^1\\
N_R
\end{array}
\right)\label{eq:amp3},
\end{align}
\end{enumerate}
which gives us Equation \ref{eq:aLL}.
\item[Final State $\ket{k-1,k,0}$]
This corresponds to a walk of the form $LR\dots LR$.
Since the 
sequence begins with an $L$ and ends with an $R$, we have $C_R=C_L$, which 
removes the summation over $C_R$.
In this particular case, the possible values of $C_L$ run
from 2 to $N_L$.

Let us examine $N_L^1$ and $N_R^1$. For a particular value of
$C_L$, $N_L^1$ will run from 1 to $C_L$.
$N_R^1$ will run from 1 to $C_R$, i.e.\ from 1 to $C_L$.
We now restrict Expression (\ref{eq:genamp1}) 
to walks beginning with $LR$ and ending with $LR$.
In the composition of $L$s, the first cluster must be of 
size 1: $N_L^1/C_L$ is the fraction of walks having this property.
In the corresponding composition of $R$s, the last cluster must be of 
size 1: $N_R^1/C_R$ is the fraction of walks having this property.
Putting all of this together and applying Lemma \ref{lem:ones}, the sum 
(\ref{eq:genamp1}) becomes
\begin{align}
\sum_{C_L=2}^{N_L}
\sum_{\parbox[c]{1.5cm}{\tiny $N_L^1=\max (1,$\\ $2C_L-N_L)$}}^{C_L}
\sum_{\parbox[c]{1.5cm}{\tiny $N_R^1=\max (1,$\\ $2C_L-N_R)$}}^{C_L}
\frac{(-1)^{n+N_L^1+N_R^1}}{(2)^{n/2}}
\frac{N_L^1(N_R^1)}{C_L^2}
\left(
\begin{array}{c}
C_L\\
N_L^1\\
N_L
\end{array}
\right)
\left(
\begin{array}{c}
C_L\\
N_R^1\\
N_R
\end{array}
\right).\label{eq:amp5} 
\end{align}
We consider the cases
\begin{enumerate}
\item 
\underline{$N_L=C_L=N_L^1$} and \underline{$N_R=C_R=N_R^1$} :
This corresponds to an alternating sequence of $L$s and $R$s, 
$LRLRLR\dots LR$. Clearly this path only exists if $N_L = N_R$, so 
the amplitude is simply $\delta_{N_L,N_R}/2^{n/2}$.
\item 
\underline{$N_L=C_L=N_L^1$} and \underline{$N_R > C_R > N_R^1$} :
Expression (\ref{eq:amp5}) becomes
\begin{align}
\sum_{\parbox[c]{1.5cm}{\tiny $N_R^1=\max (1,$\\ $2N_L-N_R)$}}^{N_L-1}
\frac{(-1)^{N_R+N_R^1}}{(2)^{n/2}}
\frac{N_R^1}{N_L}
\left(
\begin{array}{c}
N_L\\
N_R^1\\
N_R
\end{array}
\right).\label{eq:amp4} 
\end{align}

\item 
\underline{$N_L > C_L > N_L^1$} and \underline{$N_R = C_R = N_R^1$} :
Expression (\ref{eq:amp5}) becomes
\begin{align}
\sum_{\parbox[c]{1.5cm}{\tiny $N_L^1=\max (1,$\\ $2N_R-N_L)$}}^{N_R-1}
\frac{(-1)^{N_L+N_L^1}}{(2)^{n/2}}
\frac{N_L^1}{N_R}
\left(
\begin{array}{c}
N_R\\
N_L^1\\
N_L
\end{array}
\right).
\label{eq:amp6} 
\end{align}
\item 
\underline{$N_L > C_L > N_L^1$} and \underline{$N_R > C_R > N_R^1$} :
The amplitude is
\begin{align}
\sum_{C_L=2}^{N_L-1}
\sum_{\parbox[c]{1.5cm}{\tiny $N_L^1=\max (1,$\\ $2C_L-N_L)$}}^{C_L-1}
\sum_{\parbox[c]{1.5cm}{\tiny $N_R^1=\max (1,$\\ $2C_L-N_R)$}}^{C_L-1}
\frac{(-1)^{n+N_L^1+N_R^1}}{(2)^{n/2}}
\frac{N_L^1(N_R^1)}{C_L^2}
\left(
\begin{array}{c}
C_L\\
N_L^1\\
N_L
\end{array}
\right)
\left(
\begin{array}{c}
C_L\\
N_R^1\\
N_R
\end{array}
\right).\label{eq:amp7} 
\end{align}

\end{enumerate}
\item[Final State $\ket{k+1,k,0}$]
This corresponds to a walk of the form $LR\dots RL$.
Since the 
sequence begins with an $L$ and ends with an $L$, we have $C_R=C_L-1$, which 
removes the summation over $C_R$.
In this particular case, the possible values of $C_L$ run
from 2 to $N_L$.

For a particular value of $C_L$, $N_L^1$ will run from 2 to $C_L$, while 
$N_R^1$ can run from 0 to $C_L-1$. We restrict Expression (\ref{eq:genamp1})
to walks beginning with $LR$ and ending with $RL$. This places restrictions 
on the composition of the $N_L$ $L$s, but not on the $R$s. The fraction of 
walks that begin with a single $L$ is $N_L^1/C_L$. Of these, the fraction
that end also in a single $L$ is $(N_L^1-1)/(C_L-1)$.
Putting all of this together and applying Lemma \ref{lem:ones}, the sum 
(\ref{eq:genamp1}) becomes
\begin{align}
\sum_{C_L=2}^{N_L}
\sum_{\parbox[c]{1.5cm}{\tiny $N_L^1=\max (2,$\\ $2C_L-N_L)$}}^{C_L}
\sum_{\parbox[c]{1.5cm}{\tiny $N_R^1=\max (0,$\\ $2C_L-N_R-2)$}}^{C_L-1}
\frac{(-1)^{n+N_L^1+N_R^1}}{(2)^{n/2}}
\frac{N_L^1(N_L^1-1)}{C_L(C_L-1)}
\left(
\begin{array}{c}
C_L\\
N_L^1\\
N_L
\end{array}
\right)
\left(
\begin{array}{c}
C_L-1\\
N_R^1\\
N_R
\end{array}
\right).\label{eq:amp8} 
\end{align}
We consider the cases
\begin{enumerate}
\item 
\underline{$N_L=C_L=N_L^1$} and \underline{$N_R=C_R=N_R^1$} :
This corresponds to an alternating sequence of $L$s and $R$s, 
$LRLRLR\dots RL$. Clearly this path only exists if $N_L = N_R+1$, so 
the amplitude is simply $\delta_{N_L-1,N_R}/2^{n/2}$.
\item 
\underline{$N_L=C_L=N_L^1$} and \underline{$N_R > C_R > N_R^1$} :
The summations over $C_L$ and $N_L^1$ vanish and we get
\begin{align}
\sum_{\parbox[c]{1.5cm}{\tiny $N_R^1=\max (0,$\\ $2N_L-N_R-2)$}}^{N_L-2}
\frac{(-1)^{N_R+N_R^1}}{(2)^{n/2}}
\left(
\begin{array}{c}
N_L-1\\
N_R^1\\
N_R
\end{array}
\right).\label{eq:amp9} 
\end{align}
\item 
\underline{$N_L > C_L > N_L^1$} and \underline{$N_R = C_R = N_R^1$} :
The summations over $C_L$ and $N_R^1$ vanish and we get
\begin{align}
\sum_{\parbox[c]{1.5cm}{\tiny $N_L^1=\max (2,$\\ $2N_R-N_L+2)$}}^{N_R}
\frac{(-1)^{N_L+N_L^1}}{(2)^{n/2}}
\frac{N_L^1(N_L^1-1)}{N_R(N_R+1)}
\left(
\begin{array}{c}
N_R+1\\
N_L^1\\
N_L
\end{array}
\right).\label{eq:amp10} 
\end{align}
\item 
\underline{$N_L > C_L > N_L^1$} and \underline{$N_R > C_R > N_R^1$} :
The amplitude becomes
\begin{align}
\sum_{C_L=2}^{N_L-1}
\sum_{\parbox[c]{1.5cm}{\tiny $N_L^1=\max (2,$\\ $2C_L-N_L)$}}^{C_L-1}
\sum_{\parbox[c]{1.5cm}{\tiny $N_R^1=\max (0,$\\ $2C_L-N_R-2)$}}^{C_L-2}
\frac{(-1)^{n+N_L^1+N_R^1}}{(2)^{n/2}}
\frac{N_L^1(N_L^1-1)}{C_L(C_L-1)}
\left(
\begin{array}{c}
C_L\\
N_L^1\\
N_L
\end{array}
\right)
\left(
\begin{array}{c}
C_L-1\\
N_R^1\\
N_R
\end{array}
\right).\label{eq:amp11} 
\end{align}
\end{enumerate}
\item[Final State $\ket{k-1,k,1}$]
This corresponds to a walk of the form $LR\dots RR$.
Since the 
sequence begins with an $L$ and ends with an $R$ we have $C_R=C_L$, which 
removes the summation over $C_R$.
$C_L$ runs
from 1 to $N_L$ while $N_L^1$ runs from 1 to $C_L$.
$C_R$ runs
from 1 to $N_R-1$ while $N_R^1$ runs from 0 to $C_R-1$.

For the walk to be of the required form,
\begin{itemize}
\item the composition of $N_L$ $L$s must begin with a single $L$. This gives
a factor of $N_L^1/C_L$.
\item the composition of $N_R$ $R$s must \emph{not} end
with a single $R$. This gives
a factor of $(C_R-N_R^1)/C_R$.
\end{itemize}
Applying Lemma \ref{lem:ones}, the sum 
(\ref{eq:genamp1}) becomes
\begin{align}
\sum_{C_L=1}^{N_L}
\sum_{\parbox[c]{1.5cm}{\tiny $N_L^1=\max (1,$\\ $2C_L-N_L)$}}^{C_L}
\sum_{\parbox[c]{1.5cm}{\tiny $N_R^1=\max (0,$\\ $2C_L-N_R)$}}^{C_L-1}
\frac{(-1)^{n+N_L^1+N_R^1}}{(2)^{n/2}}
\frac{N_L^1(C_R-N_R^1)}{C_L^2}
\left(
\begin{array}{c}
C_L\\
N_L^1\\
N_L
\end{array}
\right)
\left(
\begin{array}{c}
C_L\\
N_R^1\\
N_R
\end{array}
\right).\label{eq:amp12} 
\end{align}
We have the two following cases:
\begin{enumerate}
\item 
\underline{$N_L=C_L=N_L^1$} :
Expression (\ref{eq:amp12}) reduces to
\begin{align}
\sum_{\parbox[c]{1.5cm}{\tiny $N_R^1=\max (0,$\\ $2N_L-N_R)$}}^{N_L-1}
\frac{(-1)^{N_R+N_R^1}}{(2)^{n/2}}
\frac{N_L-N_R^1}{N_L}
\left(
\begin{array}{c}
N_L\\
N_R^1\\
N_R
\end{array}
\right).\label{eq:amp13} 
\end{align}
\item 
\underline{$N_L > C_L > N_L^1$} :
Expression (\ref{eq:amp12}) reduces to
\begin{align}
\sum_{C_L=1}^{N_L-1}
\sum_{\parbox[c]{1.5cm}{\tiny $N_L^1=\max (1,$\\ $2C_L-N_L)$}}^{C_L-1}
\sum_{\parbox[c]{1.5cm}{\tiny $N_R^1=\max (0,$\\ $2C_L-N_R)$}}^{C_L-1}
\frac{(-1)^{n+N_L^1+N_R^1}}{(2)^{n/2}}
\frac{N_L^1(C_L-N_R^1)}{C_L^2}
\left(
\begin{array}{c}
C_L\\
N_L^1\\
N_L
\end{array}
\right)
\left(
\begin{array}{c}
C_L\\
N_R^1\\
N_R
\end{array}
\right).\label{eq:amp14} 
\end{align}
\end{enumerate}
 \end{description}

\section{Appendix: AXIOM code}
The following is the AXIOM source code for simulating a 40-step quantum walk
with memory, and initial state $\ket{-1,0,0}$.
\begin{verbatim}
--In AXIOM, comments begin with two hyphens. This is a comment.
--coin gives zero (reflection): S: |n-1,n,0> --> |n,n-1,0>
--                              S: |n+1,n,0> --> |n,n+1,0>
--coin gives one (transmission): S: |n-1,n,1> --> |n,n+1,1>
--                               S: |n+1,n,1> --> |n,n-1,1>
--In general, comments refer to PRECEDING line(s) of code
P:=matrix([[0/1 for i in 1..400] for j in 1..400]);
--set up a 400x400 transition matrix
for i in 1..394 | (divide(i,4).remainder = 0) repeat
  P(i-2,i):=1
  P(i+5,i):=1
  P(i-2,i+1):=1
  P(i+5,i+1):=-1
  P(i+4,i+2):=1
  P(i-1,i+2):=1
  P(i+4,i+3):=1
  P(i-1,i+3):=-1
--the entries of the transition matrix: normalizatiion factor (2**0.5) left until later
START:=matrix([[0/1] for i in 1..400])
START(200,1):=1
--START is the initial state vector: Instead of starting at 0, we start at point 200,
--the mid-point of the vector
COUNT:=matrix([[0/1] for i in 1..100])
--This is just for bookkeeping: COUNT will check probabilities are normalized.
--------------------------------------------------------------------------
--The initial state |-1,0,0> corresponds to START(200,1):=1
--States are ordered: |-1,0,0>, |-1,0,1>, |1,0,0>, |1,0,1>, |0,1,0> etc.
--------------------------------------------------------------------------
Q:=matrix([[0.0 for i in 1..400] for j in 1..61]);
--we will store the wavefunction amplitudes in Q
PROB:=matrix([[0/1 for i in 1..100] for j in 1..60]);
--PROB is the actual probabilities for each final state....
for j in 1..60 repeat
--we run 60 steps of the walk
  div:=2**(j-1)
  for i in 1..400 repeat
    Q(j,i):=START(i,1)/(div**0.5)
  c:=0/1
  for k in 1..99 repeat
    dummy:=4*k
    PROB(j,k):= (START(dummy,1)**2 + START(dummy+1,1)**2 + START(dummy+2,1)**2 + START(dummy+3,1)**2)/div
--Calculate sum of squares of the amplitudes to get probabilities. 4 amplitudes contribute
--to the probability at each point.
    c:=c+PROB(j,k)
--c should be a vector of 1s, if our probabilities are normalized.
  COUNT(j,1):=c
  START:=P*START
  PROB 
--Finally, display the probabilities.
\end{verbatim}
\textit{}
\bibliographystyle{plain}
\bibliography{b2}

\begin{thebibliography}{10}

\bibitem{ambainis}
Andris Ambainis.
\newblock Quantum walks and their algorithmic applications.
\newblock {\em International Journal of Quantum Information}, 1:507--518, 2003.

\bibitem{watrous}
Andris Ambainis, Eric Bach, Ashwin Nayak, Ashvin Vishwanath, and John Watrous.
\newblock One dimensional quantum walks.
\newblock In {\em Proceedings of STOC'01}, pages 37--49, 2001.

\bibitem{coins}
Todd~A. Brun, Hilary~A. Carteret, and Andris Ambainis.
\newblock Quantum walks driven by many coins.
\newblock {\em Phys. Rev. A}, 67(052317), 2003.

\bibitem{flitney}
A.~P. Flitney, D.~Abbott, and N.~F. Johnson.
\newblock Quantum random walks with history dependence.
\newblock {\em J. Phys. A: Math. Gen.}, 37:7581--7591, 2004.

\bibitem{inui}
N.~Inui and N.~Konno.
\newblock Localization of multi-state quantum walk in one dimension.
\newblock {\em Physica A}, 353:133--144, 2005.

\bibitem{segawa}
N.~Inui, N.~Konno, and E.~Segawa.
\newblock One-dimensional three-state quantum walk.
\newblock {\em Physical Review E}, 72(056112), 2005.

\bibitem{axiom}
R.~J. Jenks and R.~S. Sutor.
\newblock {\em Axiom - The Scientific Computation System}.
\newblock Springer Verlag New York, 1992.
\newblock ISBN 0-387-97855-0.

\bibitem{kempe}
Julia Kempe.
\newblock Quantum random walks - an introductory overview.
\newblock {\em Contemporary Physics}, 44(4):307--327, 2003.

\bibitem{kendon}
Viv Kendon.
\newblock Decoherence in quantum walks - a review.
\newblock {\em Math. Struct. in Comp. Sci}, 17(6):1169--1220, 2006.

\bibitem{konno}
N.~Konno.
\newblock {\em Quantum walks. In:Quantum Potential Theory, Franz, U., and
  Sch\"urmann, M., Eds.}, volume 1954 of {\em Lecture Notes in Mathematics},
  pages 309--452.
\newblock Springer-Verlag, Heidelberg, 2008.

\bibitem{konno1}
N.~Konno.
\newblock Localization of an inhomogeneous discrete-time quantum walk on the
  line.
\newblock {\em Quantum Information Processing}, In Press.

\bibitem{merlini}
Donatella Merlini, Francesca Uncini, and M.~Cecilia Verri.
\newblock A unified approach to the study of general and palindromic
  compositions.
\newblock {\em INTEGERS: Electronic Journal of Combinatorial Number Theory},
  4(A23), 2004.

\bibitem{meyer}
David~A. Meyer.
\newblock From quantum cellular automata to quantum lattice gases.
\newblock {\em J. Stat. Phys}, 85:551--574, 1996.

\bibitem{nayak}
Ashwin Nayak and Ashvin Vishwanath.
\newblock Quantum walk on the line.
\newblock {\em quant-ph/0010117}, 2000.

\bibitem{stang}
J.~B. Stang, A.~T. Rezakhani, and B.~C. Sanders.
\newblock Correlation effects in a discrete quantum random walk.
\newblock {\em J. Phys. A: Math. Theor.}, 42(175304), 2009.

\bibitem{venegas}
S.~E. Venegas-Andraca.
\newblock {\em Quantum Walks for Computer Scientists}.
\newblock Morgan and Claypool, 2008.

\end{thebibliography}
\end{document}